\documentclass[envcountsame,11pt]{article}

\usepackage{amsfonts}
\usepackage{amsmath,amssymb,amsthm,enumerate}
\usepackage{latexsym}
\usepackage{graphicx}
\usepackage[USenglish]{babel}
\usepackage{algorithm, algorithmic} % zur Darstellung of Pseudocode

\usepackage{a4wide}

\usepackage{times}

\def\squishlist{\setlength{\itemsep}{0pt}\setlength{\parsep}{0pt}%
  \setlength{\topsep}{0pt}\setlength{\partopsep}{0pt}\setlength{\parskip}{0pt}}

\let\realbfseries=\bfseries
\def\bfseries{\realbfseries\boldmath}

\newtheorem{theorem}{Theorem}
\newtheorem{definition}{Definition}
\newtheorem{lemma}{Lemma}
\newtheorem{corollary}{Corollary}

\newcommand{\usec}[1]{\smallskip\par\noindent\textbf{#1}.\ }
\newenvironment{keywords}{
       \list{}{\advance\topsep by0.35cm\relax\small
       \leftmargin=1cm
       \labelwidth=0.35cm
       \listparindent=0.35cm
       \itemindent\listparindent
       \rightmargin\leftmargin}\item[\hskip\labelsep
                                     \bfseries Keywords:]}
     {\endlist}

\newcommand{\Oof}{\mathcal{O}}

\newcommand{\NP}{\ensuremath{\mathsf{NP}}}

\newcommand{\N}{\mathbb{N}}
\newcommand{\R}{\mathbb{R}}
\newcommand{\set}[1]{\{#1\}}
\newcommand{\setext}[2] { % Mengenklammern außen und Strich in die Mitte setzen
	\left \{ #1 \, \middle | \, #2 \right \}
}
\newcommand{\tw}{\operatorname{tw}}
\newcommand{\aftw}[1]{{#1}\textnormal{-}\operatorname{tw}}
\newcommand{\afhtw}[1]{{#1}\textnormal{-}\operatorname{htw}}
\newcommand{\ftw}{\aftw{f}}
\newcommand{\fhtw}{\afhtw{f}}
\newcommand{\fwidth}{f\textnormal{-}\operatorname{width}}
\newcommand{\fw}{f\textnormal{-}\omega}

\newcommand{\assoc}{associated to }
\newcommand{\pmc}{potential maximal clique }
\newcommand{\mintri}{minimal triangulation }
\newcommand{\wfunc}{width function }
\newcommand{\ftwidth}{$f$-width }

\newcommand{\fhw}{\text{fhw}}
\newcommand{\ghw}{\text{ghw}}

\newcommand{\vrt}[1]{V(#1)}

\newcommand{\vrtH}{\vrt{H}}
\newcommand{\vrtG}{\vrt{G}}
\newcommand{\vrtT}{\vrt{T}}
\newcommand{\edge}[1]{E(#1)}

\newcommand{\edgeH}{\edge{H}}

 \newcommand{\BBB}{\mathcal{B}}
\newcommand{\CCC}{\mathcal{C}} 
 \newcommand{\FFF}{\mathcal{F}}

\newcommand{\KKK}{\mathcal{K}}

 \newcommand{\TTT}{\mathcal{T}}

\newcommand{\comment}[1]{}%{\framebox{$\to$}\marginpar{\footnotesize #1}}

\newcommand{\lukas}[1]{}

\newcommand{\gaif}[1]{\underline{#1}}

\begin{document}

%\title{Exactly computing hypertree-width measures as fast as tree-width}

\title{Computing hypergraph width measures exactly}

\author{Lukas Moll \\
Humboldt Universit\"at zu Berlin \\
Berlin, Germany
\and 
Siamak Tazari\footnote{supported by a fellowship within the Postdoc-Programme of the German Academic Exchange Service (DAAD).} \\
Massachusetts Institute of Technology\\
Cambridge, USA
\and 
Marc Thurley\footnote{supported in part by Marie Curie Intra-European Fellowship 271959 at the Centre de Recerca Matem\`{a}tica, Bellaterra, Spain}\\
Centre de Recerca Matem\`{a}tica\\
Bellaterra, Spain
}

\maketitle 

%\authorrunning{L. Moll, S. Tazari, M. Thurley\footnotemark}

%\institute{Humboldt Universit\"at zu Berlin \and University of California, Berkeley\\
%  \email{moll,tazari@informatik.hu-berlin.de\\
%  \email{thurley@eecs.berkeley.edu}} }

\maketitle

%\footnotetext{supported by a fellowship within the Postdoc-Programme of the German Academic Exchange Service (DAAD).}
\begin{abstract}
Hypergraph width measures are a class
of hypergraph invariants important in studying the
complexity of constraint satisfaction problems (CSPs).
We present a general exact exponential algorithm for a large variety of
these measures. 
A connection between these and
tree decompositions is established. This enables us to almost
seamlessly adapt the combinatorial and algorithmic results known for
tree decompositions of graphs to the case of hypergraphs and obtain
fast exact algorithms.

As a consequence, we provide algorithms which, given a hypergraph $H$
on $n$ vertices and $m$ hyperedges, compute the generalized hypertree-width of $H$ in time  $O^*(2^n)$ and
compute the fractional hypertree-width of $H$ in time  $O(1.734601^n \cdot m)$.
\footnote{We follow usual practice in omitting factors polynomial in $n$ 
each time the base of the exponent is rounded. Recall that this is justified as $c^n\cdot n^{O(1)} = O((c+\epsilon)^n)$
for every $\epsilon > 0$. In all other situations we use the notation $O^*$ which suppresses polynomial factors}
\end{abstract}

\begin{keywords}
generalized hypertree-width, fractional hypertree-width, exact exponential algorithms,
monotone $f$-width
\end{keywords}

\section{Introduction}
\label{sec:introduction}
\newcommand{\csp}{\text{CSP}} Hypergraph width measures form a class
of hypergraph invariants which play an important role in studying the
complexity of constraint satisfaction problems (CSPs).  For a set of
variables $V$, a domain $D$ and a set $C$ of constraints these
problems ask for an assignment of values in $D$ to the variables such
that each constraint is satisfied. This forms a generic framework for
many import combinatorial problems. Therefore, quite unsurprisingly,
constraint satisfaction problems are generally $\NP$-hard.  In order
to obtain a more detailed picture of the complexity of these problems,
there are at least two common directions to follow.  One of these is
the restriction of the type of constraints allowed 
(see, for example \cite{Sch78,FedVar98,Bul06,BarKoz09}).

The second direction is the restriction of the structure which
constraints impose on the variables. With a strong motivational
background in database theory, this kind of restrictions forms the
origin of hypergraph width measures
\cite{GotLeoSca02,GroMar06,Marx09a,Marx10}.  The \emph{hypergraph} of
an instance $(V,D,C)$ of a constraint satisfaction problem has
vertex set $V$ and contains for each constraint a hyperedge with the
variables occurring in this constraint. In this way we can give a
precise meaning to the restriction of the structure. For some class
$\mathcal H$ of hypergraphs, the input is restricted to instances
whose hypergraphs are contained in $\mathcal H$.

Let $\csp(\mathcal H)$ denote the constraint satisfaction problem
restricted as described above. Hypergraph width measures allow for
identification of tractable variants of this problem. In the case of
\emph{bounded arities} -- that is, when the cardinality of the hyperedges is
bounded by a constant -- it turns out that \emph{bounded tree-width}
completely describes this setting, as then $\csp(\mathcal{H})$ is
polynomial time computable if and only if $\mathcal H$ has bounded
tree-width \cite{GroSchSeg01,Grohe07b}.\footnote{Note that this
holds under the Parameterized Complexity assumption $FPT \neq W[1]$.}

In the unbounded arity case the situation is different. Several
hypergraph width measures have been identified which lead to larger
classes of tractable $\csp(\mathcal{H})$.  We have here the notion of
bounded \emph{(generalized)} \emph{hypertree-width} \cite{GotLeoSca02}
which extends bounded tree-width.  Even more general are classes
$\mathcal{H}$ of bounded \emph{fractional hypertree-width}
\cite{GroMar06} which still give rise to polynomial-time computable
constraint satisfaction problems.

\usec{Our Work} The central aim of the present work is an exact
algorithm for fractional hypertree-width. Note that there is a recent
algorithm which approximates fractional hypertree-width \cite{Marx09}
in polynomial time provided that it is constant.  But not only is this
algorithm unsuitable for large fractional hypertree-width.  There is
also no known non-trivial exact algorithm for this problem.

We remedy this situation by presenting an algorithm that more
generally computes any hypertree-width measure defined by some
\emph{monotone width function} $f$.  This implies an algorithm
for both fractional and generalized hypertree-width by essentially the
same means. We achieve this by reducing the problem to
computing a minimal triangulation of the underlying Gaifman graph of
the given hypergraph. Indeed, we show that it is sufficient to compute
a tree decomposition of the Gaifman graph while measuring the width of
sets of vertices in the given hypergraph. This enables us to almost
seamlessly adapt the combinatorial and algorithmic results known for
tree decompositions of graphs to the case of hypergraphs and obtain
fast exact algorithms. 

\begin{theorem}
  Let  $H$ be a hypergraph on $n$ vertices and $m$ hyperedges.
  \begin{enumerate}[(i)]
   \item The generalized hypertree-width of $H$ can be computed in time $O^*(2^n)$.
   \item The fractional hypertree-width of $H$ can be computed in time $O(1.734601^n \cdot m)$.
\end{enumerate}
\end{theorem}
The central idea of this algorithm is the adaptation of the algorithm in 
\cite{FomKraTodVil08} and the results of \cite{FomVil08} to the situation of 
hypergraphs. All of these algorithms require exponential space in the worst case. 
The proof of this result is presented in Section~\ref{sect:fwidth}.

\section{Preliminaries}
\label{sect:prelim}

\usec{Graphs and Hypergraphs} A \emph{hypergraph} is a pair $H =
(V(H),E(H))$ consisting of a set of \emph{vertices} $V(H)$ and a set
$E(H)$ of subsets of $V(H)$, the \emph{hyperedges} of $H$. Two
vertices are \emph{adjacent} if there exists an edge that contains
both of them. Unless otherwise mentioned, our hypergraphs have $n$
vertices and $m$ edges and do not contain isolated vertices (i.e.
vertices which do not occur in an edge of $H$). 

A \emph{graph} is a hypergraph in which every hyperedge has
cardinality $2$. Thus every concept defined for hypergraphs is also
given for graphs; however, there will be some notions we will use for
graphs exclusively.  For a subset $U \subseteq \vrtG$, we write $G[U]$
to denote the \emph{subgraph} of $G$ induced by $U$. Furthermore, $G -
U$ denotes the graph $G[\vrtG \setminus U]$. The \emph{neighborhood}
of a vertex $v \in \vrtG$ is $N(v) = \{ u \mid \{u,v\} \in E\}$; this
extends to sets of vertices by defining $N(S) = \bigcup_{v\in S} N(v)
\setminus S$.  A \emph{clique} of $G$ is a set $C \subseteq \vrtG$ such
that all vertices in $C$ are pairwise adjacent in $G$. A clique is
\emph{maximal} if it is not properly contained in another clique.  For
a set $S \subseteq V$ we define $S^2 = \{ \{u,v\} \mid u,v \in S,\,
u\neq v\}$.
The \emph{Gaifman graph} or \emph{primal graph} of a hypergraph $H$ is
the graph $\gaif H$ on $\vrtH$ with $E(\gaif H) := \{\{u, v\} \mid u,v
\in e,\, \text{ for some } e \in E(H)\}$.

\usec{Tree Decompositions and Width Functions} 
A \emph{tree decomposition} of a hypergraph $H$ is a pair $(T,\BBB)$,
where $T$ is a tree and $\BBB=\set{B_t \mid t \in \vrtT}$ is a family of
subsets of $\vrtH$, called \emph{bags}, such that
\begin{itemize} \squishlist
 \item[(i)] every vertex of $H$ appears in some bag of $\BBB$;
 \item[(ii)] for every hyperedge $e \in E(H)$ there is a $t \in \vrtT$
   such that $e \subseteq B_t$; and
 \item[(iii)] for every vertex $v \in \vrtH$ the set of bags
   containing $v$ forms a subtree $T_v$ of $T$.
\end{itemize}
A \emph{\wfunc}on the vertex set $V$ is a monotone function $f : 2^{V}
\rightarrow \R_0^+$, i.e.\ with $f(X) \leq f(Y)$ for $X \subseteq Y$.
We define $\FFF(V)$ to be the set of all width functions on $V$.  The
\emph{$\fwidth$ of a tree decomposition $\TTT$} is $\max\{f(B_t) \mid
t \in \vrtT\}$.  The \emph{$f$-hypertree-width of a hypergraph $H$},
denoted by $\fhtw(H)$, is the minimum $\fwidth$ of all tree
decompositions of $H$. We call such a tree decomposition an
\emph{$f$-optimal tree decomposition}. When considering graphs, we use
the analogous notion of \emph{$f$-tree-width} and denote it by
$\ftw(G)$.  In this setting we obtain the \emph{tree-width} of a
hypergraph $H$ as follows.

\begin{definition}
Let $s(X) = \vert X \vert - 1$; then the tree-width of $H$ is 
$\tw(H):= \afhtw{s}(H)$. 
\end{definition}
Similarly, we can define other well-known width measures. Let $H$ be a
hypergraph and $X \subseteq V(H)$. An \emph{edge cover} (w.r.t.\ $H$)
of $X$ is a subset $E' \subseteq E(H)$ such that $X \subseteq
\bigcup_{e \in E'} e$. Define $\rho_H(X)$ as the size of the smallest
edge cover of $X$ w.r.t.\ $H$. Note that this number is well-defined,
as $H$ does not contain isolated vertices.

Relaxing this, we arrive at \emph{fractional edge covers}. For a set
$X \subseteq V(H)$ a mapping $\gamma: E(H) \to [0,1]$ is a fractional
edge cover of $X$ (w.r.t.\ $H$), if $ \sum_{v \in e} \gamma(e) \ge 1$
for all $v \in X$.  Then $\rho^*_H(X)$ is the minimum of $\sum_{e \in
  E(H)} \gamma(e)$ taken over all fractional edge covers of $X$
w.r.t.\ $H$.
\begin{definition}
Let $H$ be a hypergraph.
\begin{itemize}
\item[$\bullet$] The \emph{generalized hypertree-width} of $H$ is
  $\ghw(H):= \afhtw{\rho_H}(H)$.
\item[$\bullet$] The \emph{fractional hypertree-width} of $H$ is
  $\fhw(H):= \afhtw{\rho^*_H}(H)$.
\end{itemize}
\end{definition}

\usec{Separators} For two non-adjacent vertices $u,v$ of a graph $G$,
a set $S \subseteq \vrtG$ is a \emph{$u,v$-separator} if $u$ and $v$
are in different components of $G - S$. Further, $S$ is a
\emph{minimal $u,v$-separator} if no proper subset of $S$ is a
$u,v$-separator. Generally, $S$ is a \emph{minimal separator} if it is
a minimal $u,v$-separator for some $u,v$. By $\Delta_G$ we denote the
set of all minimal separators of $G$.
Observe that a minimal separator of $G$ can be contained in another one. 
We call minimal separators not containing another one
\emph{inclusion-minimal} separators and denote the set of these 
by $\Delta_G^*$.

Let $\CCC_G(S)$ denote the set of connected components of $G - S$~(see
Fig.~\ref{fig:omegaComponents}).  A component $C \in \CCC_G(S)$ is
\emph{full} w.r.t.\ $S$, if $N(C) = S$.  By $\CCC_G^*(S)$ we denote
the set of all full connected components of $G - S$.  A \emph{block}
associated with an $S \in \Delta_G$ is a pair $(S,C)$ for some component
$C \in \CCC_G(S)$. A block is called \emph{full} if $C$ is full
w.r.t.\ $S$. Note that by definition, the set $S$ of a block $(S,C)$
is required to be a minimal separator. The \emph{realization} $R(S,C)$
of a block is the graph obtained from $G[S \cup C]$ by turning $S$
into a clique.

\usec{Triangulations, Potential Maximal Cliques}
A graph $G$ is \emph{triangulated} or \emph{chordal} if every cycle of
length at least $4$ in $G$ has a \emph{chord}, that is, an edge
between two non-consecutive vertices of the cycle. A
\emph{triangulation} of $G$ is a chordal graph $I$ on $\vrtG$ such
that $E(G) \subseteq E(I)$. Furthermore, $I$ is a \emph{minimal
  triangulation} if there is no chordal graph $I'$ on $\vrtG$ with
$E(G) \subseteq E(I') \subset E(I)$.

A set $\Omega \subseteq \vrtG$ is a \emph{\pmc}of $G$, if there is a
\mintri $I$ of $G$ such that $\Omega$ is a maximal clique in $I$.  The
set of all potential maximal cliques of $G$ is denoted by $\Pi_G$.
Let $\Omega$ be a \pmc with the components $\CCC(\Omega)$ of $G -
\Omega$ and $C \in \CCC(\Omega)$; then $(N(C), C)$ is called a
\emph{block associated with $\Omega$}~(see
Fig.~\ref{fig:omegaComponents}).

Finally, we define the \emph{$f$-clique-number of $G$} to be
	$\fw(G) := \max_{\textnormal{clique } \Omega \textnormal{ of } G} \quad f(\Omega)$.

\begin{figure}[t]
	\centering
	\begin{minipage}[t]{0.495\textwidth}
		\centering{\includegraphics[scale=0.5, trim = 10 30 10 5,clip]{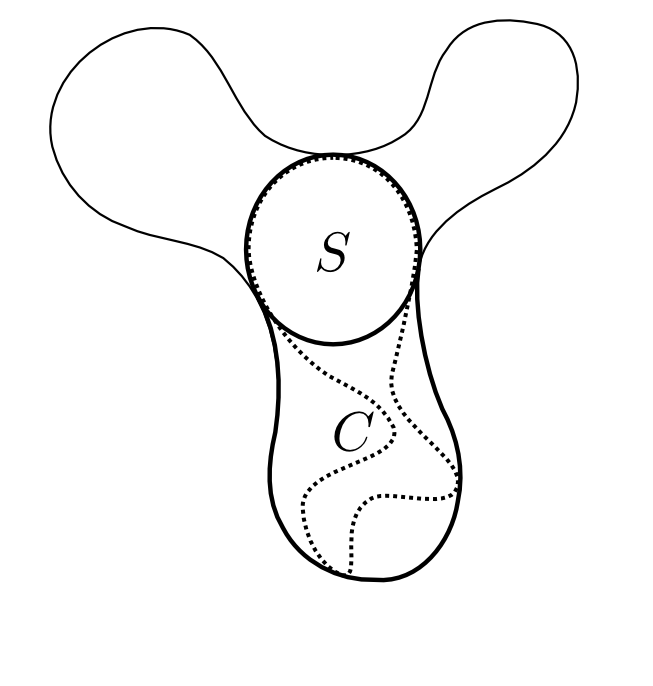}} % zum Beschneiden: [scale=0.75,trim= 10 0 10 0,clip]
	\end{minipage}
	\begin{minipage}[t]{0.495\textwidth}
		\centering{\includegraphics[scale=0.5, trim = 10 30 10 5,clip]{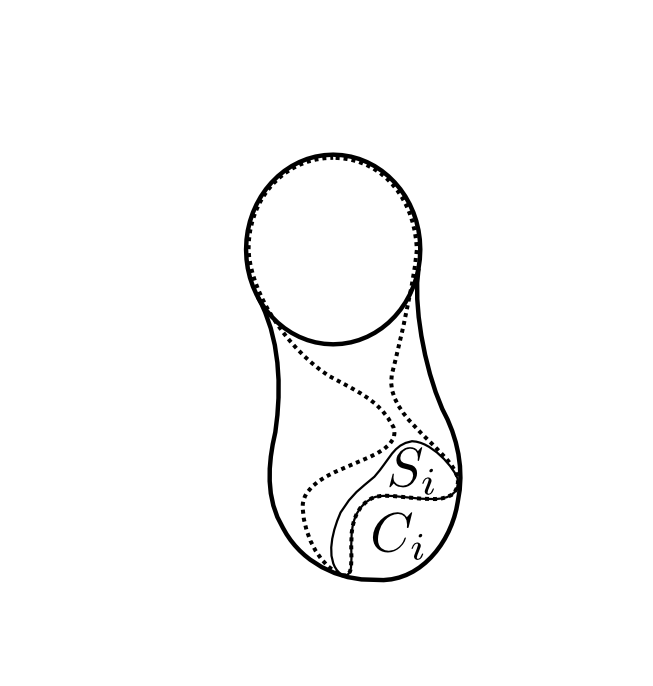}}
	\end{minipage}
	\caption{Left: The graph $G$ with separator $S$, component $C
          \in \CCC_G(S)$ of $G - S$ and $\Omega$ (dotted line).
	  Right: The block (S,C) with $\Omega$ (dotted line) and the full
          block $(S_i,C_i)$ associated with $\Omega$ with
          $C_i \in \CCC_G(S_i) \cap \CCC_G(\Omega)$ and $S_i = N(C_i)$.}
	\label{fig:omegaComponents}
\end{figure}

\section{Computing $f$-Optimal Tree Decompositions of Graphs and
  Hypergraphs}
\label{sect:fwidth}
A tree decomposition $(T,\BBB)$ with $\BBB=\setext{B_t}{t \in \vrtT}$
is \emph{small}, if for all $t,t' \in \vrtT$ with $t \neq t'$ we have
$B_t \nsubseteq B_{t'}$. We need some well-known facts about tree
decompositions:

\begin{lemma}\label{lem:well-known}
  Let $G$ be a graph, $\TTT = (T, (B_t)_{t \in V(T)})$ a tree
  decomposition of $G$, and $f \in \FFF(\vrtG)$ a width function. Then
  the following holds.
\begin{itemize} \squishlist
\item[(i)] For every clique $\Omega \subseteq V$ in $G$ there is a $t
  \in V(T)$ such that $\Omega \subseteq B_t$.
\item[(ii)] There is a small tree decomposition $\TTT'$ such that
  $\fwidth(\TTT') = \fwidth(\TTT)$.
\item[(iii)] For all $s,t,t' \in \vrtT$ such that $t'$ lies on the path
  from $s$ to $t$ in $T$, we have $B_{s} \cap B_{t} \subseteq B_{t'}$.
\end{itemize}
\end{lemma}
It is important to note here, that we will use the notion of
$\afhtw{f}$ in a slightly unusual way. Similarly to the functions
$\rho_H$ and $\rho^*_H$, we will be interested in some width function
$f_H$ which is defined on a hypergraph $H = (V(H),E(H))$, but then
apply it to tree-decompositions of a \emph{graph} $G$. The sole
prerequisite is here, that $V(G) = V(H)$ to ensure that
$\aftw{f_H}(G)$ is still well-defined.  It turns out that this very
concept of measuring the width of a tree decomposition of a graph
using the width function defined on a given hypergraph is the crucial
idea that makes our algorithm work in such a general form. We will
make the dependence of $f$ on $H$ explicit by the subscript $f_H$,
whenever this is important.

\begin{lemma} \label{MyLem01}
Let $H$ be a hypergraph, $\gaif H$ its Gaifman graph and $f_H$ a
width function on $V(H)$. Then
$$
\afhtw{f_H}(H) = \aftw{f_H}(\gaif{H}) \, .
$$
In particular, $\TTT$ is a tree decomposition of $H$ 
if, and only if, it is a tree decomposition of $\gaif{H}$.
\end{lemma}
\begin{proof} 
  It is easy to see that any tree decomposition of $H$
  is a tree decomposition of $\gaif H$.  Conversely, the fact that any
  tree decomposition of $\gaif H$ is a tree decomposition of $H$
  follows from Lemma~\ref{lem:well-known}~(i).
\end{proof}

Let $G$ be a graph and $\KKK_G$ the set of maximal cliques of $G$.
The labeled tree ${\TTT := {(T, (\Omega_t)_{t \in V(T)})}}$ is a
\emph{tree on $\KKK_G$}, if every maximal clique of $\KKK_G$
corresponds to exactly one vertex of $T$.  $\TTT$ is a
\emph{clique-tree of $G$}, if it satisfies the
\emph{clique-intersection property}:
\begin{itemize}
 \item[\textup{\textbf{(CI)}}] For every pair $\Omega, \Omega'
   \in \KKK_G$ of distinct cliques $\Omega \cap \Omega'$ is contained in every
   clique on the unique path connecting $\Omega$ and $\Omega'$ in $\TTT$.
\end{itemize}
It is well known (see e.g. Theorem 3.1 in \cite{BlaPey93}) that a graph $G$
is chordal if and only if it has a clique tree.

\begin{lemma} \label{lemX1} 
  Let $G$ be a chordal graph and $f \in \FFF(\vrtG)$ a width function.
  Then
  \begin{displaymath}
    \ftw(G) = \fw(G)
  \end{displaymath}
\end{lemma}
\begin{proof}
  Let $\Omega$ be a clique of $G$ that maximizes $f(\Omega)$.  By
  Lemma~\ref{lem:well-known}~(i) every tree decomposition of $G$
  contains a bag that contains $\Omega$. This proves $\ftw(G) \geq
  f(\Omega) = \fw(G)$.
  To see $\ftw(G) \le \fw(G)$, let $\TTT$ be a clique-tree of $G$.
  Clearly, $\TTT$ is a tree decomposition of $G$ with $\fwidth(\TTT) =
  \fw(G)$.
\end{proof}

\begin{lemma} \label{lemX2} 
Let $G$ be a graph and $f \in \FFF(\vrtG)$ a width function.
Then
\begin{equation}\label{eq:mintriang}
\ftw(G) = \min_{\substack{\textnormal{triangulation} \\ I \textnormal{ of } G}} \quad \fw(I) \, .
\end{equation}
Furthermore, the minimum on the right-hand side is attained by a minimal triangulation of $G$.
\end{lemma}
\begin{proof}
Let $I$ be any triangulation of $G$.  Since $E(G) \subseteq E(I)$,
every tree decomposition of $I$ is also a tree decomposition of $G$
and so, $\ftw(G) \leq \ftw(I)$. By Lemma~\ref{lemX1}, we have thus $\ftw(G)
\le \fw(I)$.

For the other direction, let $\TTT = (T, (B_t)_{t \in V(T)})$ be 
a small $f$-optimal tree decomposition of $G$, i.e.\ $\fwidth(\TTT) = \ftw(G)$. 
We construct a triangulation $I := (V(G), E(I))$ of $G$ by transforming the vertices 
of every bag of $\TTT$ into a clique in $I$. That is,
%\begin{displaymath}
  ${E(I) := \setext{\{v,u\}}{v \neq u, \enspace \exists t \in V(T): v,u \in B_t}}$.
%\end{displaymath}
Obviously $\TTT$ is still a tree decomposition of 
$I$ with $\fwidth(\TTT) = \fw(I)$. We show that $I$ is chordal by 
arguing that $\TTT$ is a clique-tree of $I$.
To see this, note that Lemma~\ref{lem:well-known}~(i) and the fact 
that $\TTT$ is small imply that there
is a bijection between maximal cliques of $I$ and bags of $\TTT$.
The clique-intersection property holds by Lemma~\ref{lem:well-known}~(iii). 
The monotonicity of $f$ implies that the triangulation $I$
that minimizes the right-hand side of~(\ref{eq:mintriang}) can be
chosen to be minimal.
\end{proof}

\subsection{An Algorithm to Compute the $f$-tree-width of Graphs}

The following facts about minimal separators and potential maximal
cliques are well-known, see e.g.\ Theorem~2.10 in~\cite{KloKraSpi97}
and Lemma~3.14 in~\cite{BouTod01}:

\begin{lemma}\label{lem:well-known2}
Let $G$ be a graph, $I$ a minimal triangulation of $G$, and $\Omega$ a
\pmc of $G$.
\begin{enumerate}[(i)] \squishlist
\item\label{lem:wk2:incminsep} Every block associated with an
  inclusion-minimal separator $S$ of $G$ is a full block,
  i.e.\ $\CCC_G(S) = \CCC_G^*(S)$.
\item\label{lem:wk2:minsepi} Every minimal separator of $I$ is also a
  minimal separator of $G$, i.e.\ $\Delta_I \subseteq \Delta_G$.
\item\label{lem:wk2:omegablock} Every block $(S,C)$ \assoc $\Omega$
  is, in fact, a full block of $G$; in particular, $S \in \Delta_G$.
\end{enumerate}
\end{lemma}
We proceed with a lemma from \cite{KloKraSpi97}:
\begin{lemma}[Lemma~3.1 in \cite{KloKraSpi97}] \label{lem4.1} 
  Let $G$ be a graph, $S$ a minimal separator of $G$, and $I_C$ a
  \mintri of $R(S,C)$ for each component $C$ of $G - S$.  Then the graph
  $I$ on $\vrtG$ with $E(I) := \bigcup_{C \in \CCC_G(S)} E(I_C)$ is a
  \mintri of $G$.

  Conversely, let $I$ be a \mintri of $G$ and $S$ a minimal separator of
  $I$.  Then $I[S \cup C]$ is a \mintri of $R(S,C)$ for each component
  $C$ of $G - S$.
\end{lemma}
The following lemma is an extension of Theorem~3.2. in \cite{KloKraSpi97} to our
situation.
\begin{lemma} \label{paper2_kor3.1} 
	Let $G$ be a non-complete graph and $f \in \FFF(\vrtG)$ a width function.
	Then
	\begin{equation}\label{eq:minsepalg}
		\ftw(G) = \min_{S \in \Delta_G} \quad \max_{C \in \CCC_G(S)} \ftw(R(S,C)).
	\end{equation}
\end{lemma}
\begin{proof}	
	Let $S \in \Delta_G$ be any minimal separator of $G$.  For
        every component $C \in \CCC_G(S)$, let $I_C$ be a minimal
        triangulation of $R(S,C)$ with \mbox{$\ftw(R(S,C)) =
          \fw(I_C)$} as guaranteed by Lemma~\ref{lemX2}.  By
        Lemma~\ref{lem4.1} the graph $I$ on $\vrtG$ with $E(I) :=
        \bigcup_{C \in \CCC_G(S)} E(I_C)$ is a \mintri of $G$.  By
        construction, there can not be an edge in $I$ connecting two
        different components in $\CCC_I(S) = \CCC_G(S)$; also, $S$ is
        a clique in $I$ and in each $I_C$. Thus for every clique
        $\Omega$ of $I$ there is a component $C \in \CCC_G(S)$ with
        $\Omega \subseteq S \cup C$ and $\Omega$ is also a clique of
        $I_C$.  We have thus
	\begin{displaymath}
		\ftw(G)
		\leq \fw(I)
		= \max_{C \in \CCC_G(S)} \fw(I_C)
		= \max_{C \in \CCC_G(S)} \ftw(R(S,C)) \, ,
	\end{displaymath}
	where the left most inequality is given by Lemma~\ref{lemX2}.
	
	Conversely, let $I$ be a minimal triangulation of $G$ that
        minimizes $\fw(I)$ and hence \mbox{$\ftw(G) = \fw(I)$} by
        Lemma~\ref{lemX2}.  Let $S$ be a minimal separator of $I$; by
        Lemma~\ref{lem:well-known2}~(\ref{lem:wk2:minsepi}), we know
        $S \in \Delta_G$. By Lemma~\ref{lem4.1}, we have that $I[S \cup C]$
        is a \mintri of $R(S,C)$ for every component $C \in \CCC_G(S)$
        of $G - S$.  Since every clique of $I[S \cup C]$ is also a
        clique of $I$, we have
	\begin{displaymath}
		\ftw(R(S,C))
		\leq \fw(I[S \cup C])
		\leq \fw(I)
		= \ftw(G).
	\end{displaymath}
	Again, the leftmost inequality follows from Lemma~\ref{lemX2}.
\end{proof}
Lemma~\ref{paper2_kor3.1} provides an equation for the $f$-tree-width
of a graph in terms of its minimal separators. However, it would be
preferable to work only with inclusion-minimal separators and full
blocks. Fortunately, this can be achieved via the following lemma:

\begin{lemma} \label{kor3.1_new}  
  Let $G$ be a non-complete graph and $f \in \FFF(\vrtG)$ a
  width function.  Then
  \begin{displaymath}
    \ftw(G) = \min_{S \in \Delta_G^*} \quad \max_{C \in \CCC_G^*(S)} \ftw(R(S,C)).
  \end{displaymath}
\end{lemma} 
\begin{proof}
  Suppose the minimum on the right-hand side of~(\ref{eq:minsepalg})
  is achieved only by non-inclusion-minimal separators and let $S \in
  \Delta_G$ be such a separator. Let $S' \subset S$ be an
  inclusion-minimal separator in $\Delta_G$. Consider a component $C
  \in \CCC_G(S')$; it must be that $C = S'' \cup C_1 \cup \dots \cup
  C_t$, where $S'' \subseteq S \setminus S'$ and $C_1,\dots,C_t \in
  \CCC_G(S)$. Let $\TTT_i$ be obtained from an $f$-optimal tree
  decomposition for $R(S,C_i)$, for $1 \leq i \leq t$, by removing the
  vertices of $S \setminus (S' \cup S'')$ from every bag. By creating
  a bag $B_o$ containing $S' \cup S''$ and connecting each one of
  these tree decomposition to it, we obtain a tree decomposition
  $\TTT$ for $R(S',C)$ with $\ftw(R(S',C)) \leq \fwidth(\TTT) \leq
  \max_{1 \leq i \leq t} \ftw(R(S,C_i))$. Since this is true for every
  block associated with $S'$, we obtain a contradiction, i.e.\ the
  minimum is indeed achieved by an inclusion-minimal separator
  $S'$. But then Lemma~\ref{lem:well-known2}~(\ref{lem:wk2:incminsep})
  guarantees that $\CCC_G(S') = \CCC_G^*(S')$. 
\end{proof}

\noindent It remains to show how to compute $f$-optimal tree decompositions of
full blocks. This is done in Lemma~\ref{lem4.8} below by using the
following lemma from~\cite{BouTod01} (cf.\ Fig.~\ref{fig:omegaComponents}):

\begin{lemma}[Theorem 4.7 in \cite{BouTod01}] \label{thm4.7graph} 
  Let $G$ be a graph and $(S,C)$ a full block of $G$.  Then a graph
  $I_R$ is a \mintri of $R(S,C)$ if and only if there is a \pmc
  $\Omega \subseteq S \cup C$ of $G$ with $S \subset \Omega$ such that
  the following holds:
  
  We have $V(I_R) := S \cup C$ and $E(I_R) := \bigcup_{i=1}^p E(I_i)
  \cup \Omega^2$, where $I_i$ is a \mintri of $R(S_i,C_i)$ for each
  block $(S_i,C_i)$ \assoc $\Omega$ in $R(S,C)$.
\end{lemma}

\begin{lemma} \label{lem4.8}  
  Let $G$ be a graph, $(S,C)$ a full block of $G$, and $f \in
  \FFF(\vrtG)$ a width function.  Then
  \begin{displaymath}
    \ftw(R(S,C)) = \min_{\substack{\Omega \in \Pi_G, \\ S \subset \Omega \subseteq S \cup C}} \quad \max_i \{f(\Omega), \ftw(R(S_i,C_i))\} \, ,
  \end{displaymath}
  where the maximum is taken over all blocks $(S_i,C_i)$ \assoc $\Omega$ in $R(S,C)$.
\end{lemma}
\begin{proof}
  Let $I_R$ be a minimal triangulation of $R(S,C)$, that minimizes
  $\fw(I_R)$.  By Lemma~\ref{lemX2}, we have \mbox{$\ftw(R(S,C)) =
    \fw(I_R)$}.  Lemma~\ref{thm4.7graph} implies the existence of a
  \pmc $\Omega \subseteq (S,C)$ of $G$ with $S \subset \Omega$ such
  that the following is true:
	
  For each block $(S_i,C_i)$ \assoc $\Omega$ in $R(S,C)$ there is a
  minimal triangulation $I_i$ of $R(S_i,C_i)$ such that $I_R = (S \cup
  C, E(I_R))$ with $E(I_R) := \bigcup_{i=1}^p E(I_i) \cup \Omega^2$.
  Clearly $\Omega$ is a clique in $I_R$ and hence $f(\Omega) \leq
  \fw(I_R)= \ftw(R(S,C))$.
	
  Now let $(S_i,C_i)$ be any block \assoc $\Omega$ in $R(S,C)$.  By
  definition, $S_i$ is a clique in $R(S_i,C_i)$ and therefore also in
  $I_i$ and $I_R$.  Hence, $I_R[S_i \cup C_i] = I_i$ by definition of
  $E(I_R)$.  Thus, every clique of $I_i$ is also a clique of $I_R$ and
  we have
  \begin{displaymath}
    \ftw(R(S_i,C_i))
    \leq \fw(I_i)
    \leq \fw(I_R)
    = \ftw(R(S,C)) \, .
  \end{displaymath}
  The leftmost inequality holds by Lemma~\ref{lemX2}.

  For the other direction let $\Omega$ be some potential maximal
  clique of $G$ satisfying $S \subset \Omega \subseteq S \cup C$ and
  define $w := \max_i \{f(\Omega), \ftw(R(S_i,C_i))\}$.  The
  existence of such an $\Omega$ is guaranteed by
  Lemma~\ref{thm4.7graph}.  Let $\TTT_i = (T_i, (B^i_t)_{t \in
    V(T_i)})$ be $f$-optimal tree decompositions of $R(S_i,C_i)$. Each
  $S_i$ is a clique in $R(S_i,C_i)$. Thus there is a vertex $t_i \in
  V(T_i)$ with $S_i \subseteq B^i_{t_i}$ by
  Lemma~\ref{lem:well-known}~(i).  We construct a tree decomposition
  $\TTT$ of $R(S,C)$ as the union of the tree decompositions $\TTT_i$,
  adding a new vertex $t$ and the new edges $\{t, t_i\}$.  We define
  the bag of $t$ to be $B_t := \Omega$.
	
  The sets $C_i$ are the components of $R(S,C) - \Omega$. Therefore
  the realizations $R(S_i,C_i)$ do only intersect in the sets $S_i
  \subseteq \Omega$.  Hence, every edge $e$ of $R(S,C)$ is either
  contained in $\Omega$ -- and thus in $B_t$ -- or belongs to one of
  the realizations $R(S_i,C_i)$ and so, must be contained in a bag of
  the tree decomposition $\TTT_i$. We conclude that $\TTT$ is a tree
  decomposition of $R(S,C)$ with $w = \fwidth(\TTT) \geq
  \ftw(R(S,C))$. 
\end{proof}
Combining the statements of Lemmas~\ref{kor3.1_new} and \ref{lem4.8}
we construct Algorithm~\ref{algo_ftw}. Note that
Lemma~\ref{lem:well-known2}~(\ref{lem:wk2:omegablock}) and
Lemma~\ref{kor3.1_new} justify considering only full blocks in this
algorithm.

% on page~\pageref{algo_ftw}. 
\begin{algorithm}
	\caption{$\ftw(G, f \in \FFF(\vrtG), \Delta_G, \Pi_G)$} \label{algo_ftw}
	\begin{algorithmic}[1]
		\STATE compute all full blocks $(S,C)$ and sort them by size
		\FORALL {full blocks $(S,C)$ in increasing order} \label{algo_ftw_line_2}
			\IF {$(S,C)$ is inclusion-minimal}
				\STATE $\ftw(R(S,C)) := f(S \cup C)$
			\ELSE
				\STATE $\ftw(R(S,C)) := \infty$
			\ENDIF
			\FORALL {potential maximal cliques $\Omega \in \Pi_G$ with $S \subset \Omega \subseteq (S,C)$}
				\STATE compute the full blocks $(S_i,C_i)$ associated with $\Omega$ s.t. $S_i \cup C_i \subseteq S \cup C$
				\STATE $\ftw(R(S,C)) := \min \{\ftw(R(S,C)), \max_i \{f(\Omega), \ftw(R(S_i,C_i))\}\}$
			\ENDFOR
		\ENDFOR \label{algo_ftw_line_12}
		%\STATE let $\Delta_G^*$ be the set of inclusion-minimal separators of $G$
		\STATE $\ftw(G) := \min_{S \in \Delta_G^*} \quad \max_{C \in \CCC_G^*(S)} \ftw(R(S,C))$ \label{algo_ftw_line_13}
	\end{algorithmic}
\end{algorithm}

\subsection{Runtime Analysis} 
As Algorithm~\ref{algo_ftw} is an adaptation of the algorithm presented in~\cite{FomKraTodVil08}
the runtime analysis will follow closely the analysis in that paper.
However our situation necessitates a bit of preparation.
Consider some input of the algorithm consisting of $f_H$ for some hypergraph
$H = (V(H),E(H))$ and a graph $G = (V(G),E(G))$ with $V(G) = V(H)$.
It will be convenient to separate the actual running time of the algorithm 
from the time to compute the function $f_H$ on all relevant subsets of $V(G)$. 
To this end, let a \emph{table of $f_H$ w.r.t.\ $G$} be a list 
of all inclusion minimal full blocks $(S,C)$ of $G$ and all potential
maximal cliques $\Omega$ together with
%a pointer to
the values of $f_H(S\cup C)$ and $f_H(\Omega)$, respectively, for each
of these. We obtain the following result.

\begin{theorem} \label{thmMain}  Let $H $ be a hypergraph
  with $n := |\vrtH|$, $m = |\edgeH|$, and $f_H \in \FFF(\vrtH)$ a
  width function.  Let $t(m,n)$ be an upper bound for the time needed
  to compute a table of $f_H$ w.r.t.\ the Gaifman graph $\gaif{H}$.
  Then there is an algorithm that computes $\afhtw{f_H}(H)$ together
  with an $f_H$-optimal tree decomposition of $H$ in time
  $\Oof(1.734601^n + t(m,n) + mn^2)$.
\end{theorem}
The proof of this theorem readily follows from Lemma~\ref{lemAlgo} below,
the fact that the Gaifman graph of $H$ can be computed in time
$\Oof(mn^2)$, and the following results of
\cite{FomVil08} and \cite{FomVil10}:
\begin{lemma}[\cite{FomVil08,FomVil10}]\label{lem:fomvil08}
  For every graph $G$ on $n$ vertices the following is true.  We have
  $|\Delta_G| = \Oof(1.6181^n)$ and $|\Pi_G| = \Oof(1.734601^n)$.
  Furthermore, all minimal separators and all potential maximal
  cliques can be listed in time $\Oof(1.734601^n)$.
\end{lemma}

\begin{lemma} \label{lemAlgo} 
  Let $G$ be a graph with $n := |\vrtG|$ and $f \in \FFF(\vrtG)$ a
  width function.  Given the lists of all minimal separators
  $\Delta_G$ and of all potential maximal cliques $\Pi_G$ of $G$ and
  given a table of $f$ w.r.t.\ $G$, Algorithm~\ref{algo_ftw} computes
  $\ftw(G)$ together with an $f$-optimal tree decomposition of $G$ in
  time ${\Oof(n^2 \cdot |\Delta_G| + n^3 \cdot |\Pi_G|)}$.
\end{lemma}

\begin{proof}
	W.l.o.g.\ we assume here that the graph $G$ is connected.
	Otherwise we simply run the algorithm once for each connected component of $G$.	
	The correctness of the algorithm follows easily: By
        Lemma~\ref{lem4.8} and
        Lemma~\ref{lem:well-known2}~(\ref{lem:wk2:omegablock}) the for-loop in
        the lines~\ref{algo_ftw_line_2}-\ref{algo_ftw_line_12}
        correctly computes $\ftw(R(S,C))$ for all full blocks $(S,C)$
        of $G$.  Then the \ftwidth of the graph is computed in
        line~\ref{algo_ftw_line_13} using Lemma~\ref{kor3.1_new}.
	As a table of $f$ w.r.t.\ $G$ is given, the proof of the running
        time is the same as in \cite{FomKraTodVil08}.
\end{proof}

\subsection{Computing Fractional Hypertree-Width}

\begin{lemma}
Let $H=(V(H),E(H))$ be a hypergraph with $n$ vertices and $m$ hyperedges. 
A table of $\rho^*_H$ w.r.t $\gaif H$ 
can be computed in time
$
t(m,n) = \Oof(1.734601^n \cdot m).
$
\end{lemma}
\begin{proof}
Note that we can compute $\gaif H$ from $H$ in time $\Oof(mn^2)$.
By Lemma~\ref{lem:fomvil08} we can construct a list of all minimal separators
and all potential maximal cliques of $\gaif H$ in time $\Oof(1.734601^n)$.

The list of minimal separators can be used to compute a list of all full blocks
just as has been done in \cite{FomKraTodVil08} in
$\Oof(1.734601^n)$ time.
We show how to compute the values $\rho^*_H(\Omega)$ for each potential 
maximal clique; the computation for full blocks works analogously.
For each potential maximal clique $\Omega$ in the list, we set up the linear
program 
\begin{eqnarray*}
 \text{minimize} & & \sum_{e \in E(H)} \gamma_e \\
 \text{subject to} & & \sum_{e \ni v} \gamma_e \ge 1\;\quad \text{ for all } v \in \Omega.
\end{eqnarray*}
This takes $\Oof(mn)$ time and space. By standard facts from linear programming, we know that
this program has an optimal rational solution. By a standard 
linear programming algorithm (see e.g. \cite{Kar84}) this program can be solved in time $\textup{poly}(n)\cdot m$.
\end{proof}
Combining this with Theorem~\ref{thmMain}. We obtain
\begin{corollary}
The fractional hypertree-width of a given hypergraph $H$ and a
corresponding tree decomposition can be computed in time 
$\Oof(1.734601^n \cdot m)$.
\end{corollary}

\subsection{Computing Generalized Hypertree Width}
\label{subsect:edgecover}

For a function $f : A \rightarrow B$, with $|A|=n$, we say that $f(x)$
can be computed in time $\Oof(g(n))$ to mean the time needed to evaluate $f$
\emph{once} at input $x \in A$. We say \emph{a table of $f$} can be computed
in time $\Oof(g'(n))$ if the value of $f(x)$ can be computed and
stored in a table for \emph{every} $x \in A$ in total time $\Oof(g'(n))$.

Let us fix a hypergraph $H$ on $n$ vertices and $m$ edges. In order to
compute the generalized hypertree-width of $H$ using
Theorem~\ref{thmMain}, we need to compute a table of $\rho_H$
(w.r.t.\ $\gaif H$).  This can be accomplished by a fairly
straightforward dynamic programming algorithm in time $\Oof(2^n mn)$:
build a table with an entry for every pair $(U,i)$, $1 \leq i \leq m$,
where $U$ is a subset of the vertices and $\set{e_1,\dots,e_m}$ are
the edges of the graph.  For every $(U,i)$, store the minimum-size
edge cover for $U$ that uses only edges $\set{e_1,\dots,e_i}$; this
can be easily done by considering the entries stored at $(U,i-1)$ and
$(U \setminus e_i, i-1)$. An additional factor of $n$ is needed to
look up an add $n$-bit integers.

This approach of computing $\rho_H$ has the drawback 
of being dependent on $m$, which itself
might be exponential in $n$ an thus yields an overall running time of $O(n4^n)$
in the worst case. Fortunately, as we shall see now, there is an elegant 
machinery which allows us 
to significantly improve this time bound.

\subsubsection{Faster Computation using the Fast M\"obius Transform}

Using \emph{the principle of inclusion-exclusion} and the \emph{fast
  zeta transform}, Bj\"orklund et al.~\cite{BjorklundHK09} show, for a
given set $N$ of $n$ elements and a family of its subsets, how to
count the number of $k$-covers of $N$ in time $\Oof(2^n n^2)$, where
$k$ is part of the input. This leads to an $\Oof(3^n n^2)$-time
algorithm to compute a table for $\rho_H^k$, where $\rho_H^k(U)$
denotes the number of edge covers of $U \subseteq
\vrtH$ using at most $k$ hyperedges.  We show how to improve this
running time to $\Oof(2^n n^3)$.

Let $N$ be an $n$-element set and $f: 2^N \rightarrow \R$ be a
real-valued function on the set of all subsets of $N$. The zeta
transform~\cite{Rota64} of $f$, denoted by $\hat{f} :
2^N \rightarrow \R$ is defined as
\begin{equation*}
\hat{f}(Y) = \sum_{S \subseteq Y} f(S), \qquad \text{for } Y \subseteq N \, .
\end{equation*}
The straightforward method to compute a table for the zeta transform
of $f$, i.e.\ compute $\hat{f}(Y)$ for all ${Y \subseteq N}$, requires
$\Oof(3^n)$ additions in total. However, this can be improved to
$\Oof(2^n n)$ additions using Yates's
method~\cite{Yates37,BjorklundHK09} as specified in the following
lemma; this algorithm is known as the \emph{fast M\"obius transform}
or the \emph{fast zeta transform}; we use the latter term in this
work.

\begin{lemma}[\cite{Yates37,BjorklundHK09}]\label{lem:fastzeta}
Let $N$ be a set of $n$ elements and $f: 2^N \rightarrow \N$ a
function in the range $[-M,M]$. A table for the zeta transform
$\hat{f}$ of $f$ can be computed via $\Oof(2^n n)$ additions with $\Oof(n
\log M)$-bit integers.
\end{lemma}

\noindent For a set $X \subseteq \vrtH$, define the number of edges that avoid
$X$ as $a(X)= |\setext{e \in \edgeH}{e\cap X = \emptyset}|$. Using the
principle of inclusion-exclusion, Bj\"orklund et
al.~\cite{BjorklundHK09} show

\begin{lemma}[adapted from~\cite{BjorklundHK09}]
Let $H$ be a hypergraph. For a set $U \subseteq \vrtH$, let
$\rho_H^k(U)$ denote the number of edge covers of $U$ using at most $k$
hyperedges; furthermore, let $a(U)$ be the number of hyperedges that avoid
$U$. Then we have
\begin{equation}\label{eq:edgecover}
\rho_H^k(U) = \sum_{X \subseteq U} (-1)^{|X|} a(X)^k \, .
\end{equation}
\end{lemma}

\noindent Now we are ready to state the main result of this subsection:

\begin{theorem}\label{thm:edgecover}
Let $H$ be a hypergraph on $n$ vertices and $m$ edges. Let $\rho_H^k :
2^{\vrtH} \rightarrow \N$ be the function that counts the number of
edge covers with at most $k$ hyperedges for every subset of $\vrtH$. Then a table for
$\rho_H^k$ can be computed in time $\Oof(2^n n^3)$.
\end{theorem}
\begin{proof}
First, we compute the values $a(X)$ for every $X \subseteq \vrtH$
using the idea in~\cite{BjorklundHK09}: observe that if $e(S)$ is the
indicator function telling if $S \subseteq \vrtH$ is an edge or not,
then
\begin{equation*}
a(X) = \sum_{S \subseteq \vrtH \setminus X} e(S) = \hat{e}(\vrtH \setminus X)
\end{equation*}
can be computed from the zeta transform of $e$. Hence, a table for
$g(X) := (-1)^{|X|} a(X)^k$ can be pre-computed and stored in time
$\Oof(2^n n^2)$ using Lemma~\ref{lem:fastzeta} (having accounted for
an overhead factor of $n$ for looking up and adding $n$-bit
integers). But then, Equation~(\ref{eq:edgecover}) implies that
$\rho_H^k$ is just the zeta-transform of $g$ and hence, a table for
$\rho_H^k$ can be computed in time $\Oof(2^n n^3)$ using
$\Oof(n^2)$-bit integers.
\end{proof}

For any given hypergraph $H$ and integer $k$, we can compute a table
that stores for every subset ${U \subseteq \vrtH}$ if it has an edge
cover of size at most $k$ using Theorem~\ref{thm:edgecover}. Together
with Theorem~\ref{thmMain} this implies an $\Oof(2^n n^3)$-time
algorithm to decide whether the generalized hypertree-width of a given
graph is at most $k$ and if so, compute a corresponding tree
decomposition. The tree decomposition with the \emph{minimum}
generalized hypertree-width can then be obtained by binary search on
$k$, adding only another factor of $n$ as overhead. Note, however,
that this method does not compute the actual (minimum) edge cover for
each bag of the tree decomposition; to this end, the simple dynamic
programming algorithm described in the beginning of this subsection
has to be used.

\begin{corollary}
The generalized hypertree-width of a given hypergraph $H$ and a
corresponding tree decomposition can be computed in time $\Oof^*(2^n)$. 
The minimum edge cover for every bag of the tree
decomposition can be computed in total time $\Oof^*(2^n m)$.
\end{corollary}

\section{Conclusion}
\label{sect:conclusion}

We present an algorithm that computes the $f$-width of a hypergraph
for any monotone function $f$. Apart from the overhead in computing
$f$, the algorithm works within the same time bound as the currently
fastest exact algorithms for tree-width.  As a consequence we obtain
fast exact algorithms to compute the generalized and
fractional hypertree-widths of a hypergraph. 

An important open question is whether these algorithms can be further
developed to also compute the more general hypertree width measures of
adaptive width and submodular width (see Marx \cite{Marx09a,Marx10}).

\bibliographystyle{alpha}
\bibliography{references}

\end{document}